\documentclass[12pt]{amsart}
\usepackage{amssymb,amscd}
\usepackage{verbatim}

\usepackage{xcolor}

\usepackage{amsmath,amssymb,graphicx,mathrsfs}   
\usepackage{enumerate}
\usepackage[colorlinks=true,allcolors = blue]{hyperref} 

\usepackage{tikz}
\usetikzlibrary{matrix}

\usepackage[all]{xy}

\usepackage{amssymb,amsfonts,amsthm,amsmath,calligra}
\usepackage{slashed}
\usepackage{yfonts}
\usepackage{mathrsfs,pifont}
\usepackage{float}

\usepackage[all]{xy}

\usepackage{tikz}

\usepackage{graphicx}
\usepackage{xcolor}

\usepackage{amssymb,amsfonts,amsthm,amsmath}
\usepackage[all]{xy}
\usepackage{slashed}
\usepackage{yfonts}
\usepackage{mathrsfs,pifont}

\let\frak\mathfrak

\def\>{\relax\ifmmode\mskip.666667\thinmuskip\relax\else\kern.111111em\fi}
\def\<{\relax\ifmmode\mskip-.333333\thinmuskip\relax\else\kern-.0555556em\fi}
\def\vsk#1>{\vskip#1\baselineskip}
\def\vv#1>{\vadjust{\vsk#1>}\ignorespaces}
\def\vvn#1>{\vadjust{\nobreak\vsk#1>\nobreak}\ignorespaces}

  \let\ssize\scriptstyle
\let\sssize\scriptscriptstyle

\let\Medskip\medskip
\def\medskip{\par\Medskip}
\let\Bigskip\bigskip
\def\bigskip{\par\Bigskip}

\let\Maketitle\maketitle
\def\maketitle{\Maketitle\thispagestyle{empty}\let\maketitle\empty}

\newtheorem{thm}{Theorem}[section]
\newtheorem{cor}[thm]{Corollary}
\newtheorem{lem}[thm]{Lemma}
\newtheorem{prop}[thm]{Proposition}

\newtheorem{defn}[thm]{Definition}

\newtheorem{exmp}[thm]{Example}

\theoremstyle{definition}                                  
\numberwithin{equation}{section}

\theoremstyle{definition}
\newtheorem*{rem}{Remark}
\newtheorem*{example}{Example}

\let\mc\mathcal
\let\nc\newcommand

\let\la\lambda

\let\phi\varphi

\let\om\omega

\let\der\partial

\let\geq\geqslant

\let\leq\leqslant

\let\on\operatorname
\let\bi\bibitem
\let\bs\boldsymbol

\def\C{{\mathbb C}}
\def\Z{{\mathbb Z}}

\def\F{{\mathbb F}}

\def\+#1{^{\{#1\}}}

\def\Gr{\on{Gr}}

\def\beq{\begin{equation}}
\def\eeq{\end{equation}}
\def\be{\begin{equation*}}
\def\ee{\end{equation*}}

\nc{\bea}{\begin{eqnarray*}}
\nc{\eea}{\end{eqnarray*}}
\nc{\bean}{\begin{eqnarray}}
\nc{\eean}{\end{eqnarray}}

\let\ga\gamma

\nc{\Il}{{\mc I_{\bs\la}}}
\nc{\bla}{{\bs\la}}
\nc{\Fla}{\F_\bla}
\nc{\tfl}{{T^*\Fla}}
\nc{\GL}{{GL_n(\C)}}
\nc{\GLC}{{GL_n(\C)\times\C^*}}

\let\sd s 

\def\ddk_#1{\kk_{#1}\<\>\frac\der{\der\<\>\kk_{#1}}}

\def\bul{\mathbin{\raise.2ex\hbox{$\sssize\bullet$}}}
\def\intt{\mathchoice
{\mathop{\raise.2ex\rlap{$\,\,\ssize\backslash$}{\intop}}\nolimits}
{\mathop{\raise.3ex\rlap{$\,\sssize\backslash$}{\intop}}\nolimits}
{\mathop{\raise.1ex\rlap{$\sssize\>\backslash$}{\intop}}\nolimits}
{\mathop{\rlap{$\sssize\<\>\backslash$}{\intop}}\nolimits}}

\let\kk q 
\let\cc c

\let\Ko K

\def\GZ/{Gelfand-Zetlin}
\def\KZ/{{\slshape KZ\/}}
\def\qKZ/{{\slshape qKZ\/}}
\def\XXX/{{\slshape XXX\/}}

\nc{\A}{{\mc A}}

\nc{\hsl}{\widehat{{\frak{sl}_2}}}

\nc{\BC}{{ \mathbb C}}
\nc{\lra}{\longrightarrow}
\nc{\CO}{{\mathcal{O}}}
\nc{\BZ}{{ \mathbb Z}}
\nc{\hfn}{\hat{\frak{n}}}
\nc\Zs{{\Z/p^s\Z}}
\nc\Zo{{\Zs[z]^0}}
\nc\gr{{\on{gr}}}

\nc\fD{{\frak D}}

\newcommand{\Mas}{\Phi}
\newcommand{\Ghost}{\mathsf{L}}

\newcommand{\T}{\mathsf{T}}
\newcommand{\Ver}{\mathsf{V}}

\newcommand{\matC}{\mathbb{C}}
\newcommand{\matN}{\mathbb{N}}

\newcommand{\matQ}{\mathbb{Q}}
\newcommand{\matP}{\mathbb{P}}
\newcommand{\matZ}{\mathbb{Z}}
\newcommand{\matF}{\mathbb{F}}

\newcommand{\qm}{\mathsf{QM}}

\newcommand{\bv}{\mathsf{v}}
\newcommand{\bw}{\mathsf{w}}

\newcommand{\G}{\mathsf{G}}

\usepackage{tikz}

\usetikzlibrary{decorations}
\usetikzlibrary{decorations.pathmorphing}
\usetikzlibrary{calc}

\newsavebox{\Ipm}
\savebox{\Ipm}{%
\begin{tikzpicture}[baseline= {($(current bounding box.base)-(0pt,-20pt)$)}]
\begin{scope}
\draw[dashed,line width=0.35mm] (0,0)-- (10,0) ;
\node[circle,draw,minimum size=4mm,fill=black] (c) at (0,0){};
\node[circle,draw,minimum size=4mm,fill=black] (c) at (2,0){};
\node[circle,draw,minimum size=4mm,fill=black] (c) at (4,0){};
\node[circle,draw,minimum size=4mm,fill=black] (c) at (6,0){};
\node[circle,draw,minimum size=4mm,fill=black] (c) at (8,0){};
\node[circle,draw,minimum size=4mm,fill=black] (c) at (10,0){};
\draw [-to,line width=0.5mm](0,0) -- (1.8,0);
\draw [-to,line width=0.5mm](8,0) -- (9.8,0);
\node[rectangle,draw,minimum size=4mm,fill=black] (c) at (4,-2){};
\draw [-to,line width=0.5mm](4,-2) -- (4,-0.2);
\node[rectangle,draw,minimum size=4mm,fill=black] (c) at (6,-2){};
\draw [-to,line width=0.5mm](6,-2) -- (6,-0.2);
\node at (0,0.5) { $\bv_1$};
\node at (2,0.5) { $\bv_2$};
\node at (4,0.5) { $\bv_k$};
\node at (6,0.5) { $\bv_{n-k}$};
\node at (8,0.5) { $\bv_{n-2}$};
\node at (10,0.5) { $\bv_{n-1}$};
\node at (4,-2.6) { $z_1$};
\node at (6,-2.6) { $z_2$};
\end{scope}
\end{tikzpicture}}

\begin{document}

\hrule width0pt
\vsk->

\title[The $p\,$-adic approximations of vertex functions 
via $3D$-mirror symmetry]
{The $p\,$-adic approximations of vertex functions via $3D$-mirror symmetry}

\author
[Andrey Smirnov and Alexander Varchenko]
{Andrey Smirnov$^{\star}$ and Alexander Varchenko$^{\diamond}$}

\maketitle

\begin{center}
{ Department of Mathematics, University
of North Carolina at Chapel Hill\\ Chapel Hill, NC 27599-3250, USA\/}

\end{center}

\vsk>
{\leftskip3pc \rightskip\leftskip \parindent0pt \Small
{\it Key words\/}:  Superpotentials; Vertex Functions; 
Dwork-Type Congruences, Ghosts.

\vsk.6>
{\it 2020 Mathematics Subject Classification\/}: 
\par}


{\let\thefootnote\relax
\footnotetext{\vsk-.8>\noindent
$^\star\<${\sl E\>-mail}:\enspace asmirnov@email.unc.edu
\\
$^\diamond\<${\sl E\>-mail}:\enspace  anv@email.unc.edu}}

\begin{abstract}

Using the $3D$ mirror symmetry we construct a system of polynomials $\T_s(z)$ with integral coefficients which solve 
the quantum differential equitation of $X=T^{*}\Gr(k,n)$ modulo $p^s$, where $p$ is a prime number. 
We show that  the sequence $\T_s(z)$ converges in the $p$-adic norm to the
Okounkov's vertex function of $X$ as $s\to \infty$. We prove that $\T_s(z)$ satisfy Dwork-type congruences
which lead to a new infinite product presentation of the vertex function modulo $p^s$.

\end{abstract}


\setcounter{footnote}{0}
\renewcommand{\thefootnote}{\arabic{footnote}}

\section{Introduction}

\subsection{}
The vertex functions are among the main objects studied in enumerative geometry of Nakajima's quver varieties \cite{Oko17}. These functions are analogs of  Givental's 
 $J$-functions in quantum cohomology \cite{Giv96}. The vertex functions are defined as power series
$$
\Ver(z) = \sum\limits_{d=0}^{\infty}\, c_d\, z^d  \in \matQ[[z]]
$$
where the coefficient $c_d$ counts the number of degree $d$ rational curves in a quiver variety $X$. 
More precisely, $c_d$ is given by the regularized integral of the virtual fundamental class~$\omega^{vir}$
$$
c_d: =\int\limits_{\qm_d(X,\infty)} \, \omega^{vir} 
$$
over the moduli space $\qm_d(X,\infty)$ of degree $d$ quasimaps from a rational curve $C\cong \matP^1$ to $X$ with prescribed behaviour at $\infty \in C$, see Section 7 of \cite{Oko17} for definitions. 

\subsection{}
In this paper we initiate a study of arithmetic properties of~$c_d$. 
For this goal, we consider the vertex function $\Ver(z)$ for the simplest Nakajima quiver variety, given by the cotangent bundle over the Grassmannian,  $X=T^{*}\Gr(k,n)$.

For a prime number $p$, we construct a sequence of polynomials $\T_s(z) \in \matZ[z]$, $s=0,1,\dots$, 
starting from $\T_0(z)=1$ 
which converges to the vertex function, 
$$
\lim_{s\to \infty}\, \T_s(z)=\Ver(z).
$$
The convergence is understood in the $p$-adic norm, see Theorem \ref{TtoVthm}. We refer to the polynomials $\T_s(z)$ as the {\it $p$-adic approximations of} \, $\Ver(z)$. 

We find that, unlike the vertex functions themselves, their $p$-adic approximations
satisfy a number of remarkable congruences: 
\begin{thm}[Theorem \ref{dworkth}]
{\it The $p$-adic approximations $\T_s(z)$ satisfy the Dwork-type congruences:
\bean \label{dwcong}
\dfrac{\T_{s+1}(z)}{\T_s(z^p)}=\dfrac{\T_{s}(z)}{\T_{s-1}(z^p)} \mod p^s
\eean
with $s=1,2,\dots$.} 
\end{thm}
This type of congruences played an important role in the work of Dwork \cite{Dwo69}, which laid foundation of the theory of $p$-adic hypergeometric equations.
In fact, for $X=T^{*}\matP^1$ our $T_s(z)$ are close to the truncations of the hypergeometric 
function ${}_2 F_{1}(\frac{1}{2},\frac{1}{2},1;z)$ considered by Dwork as his primary example,
but not the same.

Among other things, Theorem \ref{dworkth} implies that modulo $p^s$, the vertex function has the following infinite product presentation. 

\begin{thm}[Theorem \ref{infprodrepr}]
{\it The vertex function of 
$X=T^{*}\Gr(k,n)$ has the infinite product presentation:
 $$
\Ver(z)=\prod\limits_{i=0}^{\infty}\, \dfrac{\T_m(z^{p^{i}})}{\T_{m-1}(z^{p^{i+1}})} \mod p^m, \ \ m=1,2 \dots
 $$  
in particular, for $m=1$ we obtain
$$
\Ver(z)=\prod\limits_{i=0}^{\infty}\,\T_1(z^{p^{i}}) \mod p.
$$
 }

\end{thm}
To prove the congruences (\ref{dwcong}) we use the technique of {\it ghosts} rooted in \cite{Mel09, MeVl16} and developed further in \cite{VZ21,Var22b}. An important difference with the previous papers is that our approach here does not require working with the
whole Hasse-Witt matrices. Due to internal symmetry of the functions we consider here, only a 
 specific matrix elements of these matrices play a role. So, an alternative title of this paper could be {\it Dwork type congruences with symmetries}.

\subsection{} 
The construction of $p$-adic approximations $\T_s(z)$ is inspired by the idea of $p$-adic approximations of hypergeometric solutions of the KZ equations in \cite{SV19} and 
by the idea of $3D$-mirror symmetry, in the spirit of \cite{RSVZ19, RSVZ21}.
In Section \ref{mirrsymsect} we consider a quiver variety $X^{!}$, known as  a {$3D$-mirror} $X$. From the quiver of $X^{!}$ for a
 choice of 
a prime $p$ and $s\in \matN$ we construct a polynomial
$$
\Phi_{s}(x,z) \in \matZ[x,z].
$$
The auxiliary variables $x=(x_{i,j})$ play a role of the Chern roots of the tautological bundle
over the quiver variety $X^{!}$. The polynomial $\Mas_{s}(x,z)$ can be understood as a $p$-adic polynomial approximation of the {\it superpotential} of the $3D$-quantum field theory
with the Higgs branch $X^{!}$. We then define $\T_s(z)$ as a specific $x$-coefficient in $\Phi_{s}(x,z)$
\bean \label{cofsdef}
\T_s(z) =  \mathrm{coeff}_{x^{d p^s -1}} \Big(\Mas_s(x,z) \Big),
\eean
see Section \ref{sectionT} for details.  This definition  is natural in the sense that the operator of taking coefficients (\ref{cofsdef}) behave in many respects similar to the integration over a closed cycle in the complex setting. This operation can be viewed as an $\matF_{p^s}$ -
version of integration, see \cite{SV19, Var22a, RV21, RV22}.

 The normalized vertex function $\Ver(z)$  associated with a quiver variety, can be characterized as a unique analytic solution of the {\it quantum differential equation} which governs the quantum cohomology of $X$. For instance, for $X=T^{*} \matP^n$, $n=1,2,\dots$, these are the standard generalized hypergeometric equations.  It can be shown that the coefficient (\ref{cofsdef}) is a solution of these equations {\it modulo $p^s$}, which explains the motivation for definition~(\ref{cofsdef}).  

We also note that for our running example $X=T^{*} \matP^1$, the polynomial $\T_1(z)$ is the Hasse-Witt invariant of an elliptic curve, which was first observed to be a modulo $p$ solution to the Gauss hypergeometric differential equation by Igusa \cite{Igu58}.

\subsection{}
Among other things, congruences (\ref{dwcong}) mean that
 $I_s(z)=\T_{s+1}(z)/\T_{s}(z^p)$ is a Cauchy sequence which converges uniformly to a $\matZ_p$-valued
 analytic  function $I(z)$ in a large domain $\frak{D} \subset \matZ_p$\,. 
 That function $I(z)$ is the $p$-adic analytic continuation to $\frak D$ of the function
 $\Ver(z)/\Ver(z^p)$ defined as a ratio of power series in a neighborhood of $z=0$. 
  For points in $\frak D$ we have a modular transformation identity
 $$
z^{d} I(1/z) =  I(z)
 $$
 where $d$ is a constant depending on the choice of $p$,  see 
 Theorem \ref{ancontfunc}. This property of  $\Ver(z)/\Ver(z^p)$ differs drastically from the properties of the vertex functions
  over $\matC$, which have much more non-trivial analytic continuation.

\subsection{} 
The results of the present paper have several straightforward generalizations in the number of obvious directions. First, the quiver variety $X=T^{*}\Gr(k,n)$ which we only consider here, can be,  with some extra work, replaced by the cotangent bundles over partial flag varieties. Second, the idea of $p$-adic approximations of vertex functions can be straightforwardly applied to the vertex functions with {\it descendents}. These functions are solutions to a number of enumerative and geometric problems. For instance, as shown in \cite{Oko17} for the special choice of the descendent insertions, given by the stable envelope functions \cite{AO16, MO19}, the descendent vertex functions are equal to the {\it capping operators}. In enumerative geometry these functions count the quasimaps in $X$ with relative boundary conditions, see Section 7.4 of \cite{Oko17}. At the same time, as it was 
recently shown by Danilenko \cite{Dan22}, the capping operators can be understood as the fundamental solution matrices of the quantum Knizhnik-Zamolodchikov equations associated with mirror varieties.
Our approach suggests a natural $p$-adic approximations of all these objects. We plan to return to these ideas in separate papers.

\subsection*{Acknowledgements}

Work of A. Smirnov is partially supported by NSF grant DMS - 2054527 and by
the RSF under grant 19-11-00062.
Work of A. Varchenko is partially supported by NSF grant DMS - 1954266.

\section{Vertex functions of $T^{*}\Gr(k,n)$ \label{vertexfun}}

\subsection{}
The vertex function of the cotangent bundle over Grassmannian $X=T^{*}\Gr(k,n)$ is given by the power series:
\bean \label{verfun}
\Ver(z)=\sum_{d=0}^{\infty} \, c_{d}(u_1,\dots, u_n,\hbar) \, z^d 
\eean
with the coefficients $c_{d}(u_1,\dots, u_n,\hbar)\in \matQ(u_1,\dots,u_n,\hbar,\epsilon)$ given by:
\begin{small}
\bean \label{coeffsd}
 c_{d}(u_1,\dots, u_n,\hbar)=\sum\limits_{{d_1,\dots,d_k:}\atop {d_1+\dots+d_k=d}}\,  \Big(\prod\limits_{i,j=1}^{k} \dfrac{(\epsilon-u_i+u_j)_{d_i-d_j}}{(\hbar-u_i+u_j)_{d_i-d_j}} \Big) \Big(\prod\limits_{j=1}^{n} \prod\limits_{i=1}^{k}\, \dfrac{(\hbar+u_j-u_i)_{d_i}}{(\epsilon+u_j-u_i)_{d_i}} \Big),
\eean
\end{small}
where $(x)_d$ denotes the Pochhammer symbol with step $\epsilon$:
$$
(x)_d =\left\{ \begin{array}{rr}
x (x+\epsilon)\dots (x+(d-1) \epsilon), & d>0\\
1, & d=0\\
\dfrac{1}{(x-\epsilon)(x-2 \epsilon) \dots (x +d \epsilon)  }, & d<0
\end{array}\right. 
$$
The degree $d$ coefficient of this series counts (equivariantly) the number of degree $d$ rational curves in $X$. More precisely, it is given by the equivariant integral
\bean \label{cinteg}
c_{d}(u_1,\dots, u_n,\hbar) = \int\limits_{ [\qm_d(X,\infty)]^{\textrm{vir}}} \, \omega^{vir}
\eean 
over the virtual fundamental class on moduli space $\qm_d(X,\infty)$ of quasimaps from $\mathbb{P}^1$ to $X$, which send $\infty \in \mathbb{P}^1$ to a prescribed torus fixed point in $X$, see Section 7.2 of \cite{Oko17} for definitions.
 Using the equivariant localization,  the integral (\ref{cinteg}) reduces to the sum over the torus fixed points on $\qm_d(X,\infty)$ which gives the sum (\ref{coeffsd}). We refer to Section 4.5 of \cite{PSZ16} where this computation is done in some details.

The parameters $u_1,\dots, u_n,\hbar,\epsilon$ are the equivariant parameters of the torus $T=(\mathbb{C}^{\times})^{n}\times \mathbb{C}^{\times}_{\hbar} \times  \mathbb{C}^{\times}_{\epsilon}$ acting on the moduli space $\qm_d(X,\infty)$ in the following way: 
\begin{itemize}
    \item $(\mathbb{C}^{\times})^{n}$ acts on $W=\mathbb{C}^n$ in a natural way, scaling the coordinates with weights $u_1,\dots , u_n$. This  induces an action of $T$ on $X\cong T^{*} \Gr(k,W)$ which, in turn, induces an action of $T$ on $\qm_d(X,\infty)$. 
    \item $\mathbb{C}^{\times}_{\hbar}$ acts on $X$ by scaling the cotangent fibers with weight $\hbar$ which induces an action of $T$ on $\qm_d(X,\infty)$.
    \item $\mathbb{C}^{\times}_{\epsilon}$ acts on the source of the quasimaps $C\cong\mathbb{P}^{1}$ 
     fixing the points $0,\infty \in \mathbb{P}^{1}$. The parameter $\epsilon$ denotes the corresponding weight of the tangent space $T_{0}\, \mathbb{P}^{1}$.
\end{itemize}

\begin{exmp} In the simplest case $k=1,n=2$ corresponding to the cotangent bundle over the protective space, $X=T^{*} \mathbb{P}^1$, the vertex function (\ref{verfun}) is the  Gauss hypergeometric function:
$$
\Ver(z)={}_2 F_{1}\Big(\dfrac{\hbar}{\epsilon}, \dfrac{u_2-u_1+\hbar}{\epsilon}; \dfrac{u_2-u_1+\epsilon}{\epsilon}; z \Big).
$$
\end{exmp}

\subsection{} 

In this paper we always study the vertex function (\ref{verfun}) with the following specialization of the equivariant parameters:
\bean \label{specializ}
&& 
u_1=\dots = u_n =0,\qquad \hbar/\epsilon=\om, \qquad \om\in\C.
\eean
Later we  fix $\om$ to be a rational number, $0<\om\leq 1/2.$

In this case the coefficient (\ref{cinteg}) computes the equivariant integral in the case when  the torus $(\mathbb{C}^{\times})^{n}$ acts trivially, while $\mathbb{C}^{\times}_\hbar$ and
$\mathbb{C}^{\times}_\epsilon$ act with weights for which $\hbar/\epsilon=\om$. Since evaluation maps are proper over $\qm^{d}(X,\infty)^{\mathbb{C}^{\times}_\hbar\times \mathbb{C}^{\times}_\epsilon}$, the specialization of the vertex function at (\ref{specializ}) is well defined.

\begin{exmp} \label{examtp1}
Continuing the previous example with $k=1,\,n=2$ the specialized vertex function has the form
\bean \label{hypertp1}
\Ver(z)={}_2 F_{1}\big(\om, \om; 1; z \big) = \sum\limits_{d=0}^{\infty}\, \binom{-\om}{m}^{\!\! 2} z^d\,.
\eean
If $\om=1/2$, then the first several coefficients of this power series are :
\bean
 \label{vert12}
\Ver(z)=1+{\frac {1}{4}}z+{\frac {9}{64}}{z}^{2}+{\frac {25}{256}}{z}^{3}+{
\frac {1225}{16384}}{z}^{4}+\mc O \left( {z}^{5} \right). 
\eean
\end{exmp}

\section{$3D$-mirror symmetry and integral representations of cohomological vertex functions \label{mirrsymsect}}

\subsection{} 
Among other things, the $3$-dimensional mirror symmetry provides an integral representations of the vertex functions. To a symplectic variety $X$ this symmetry  associates a $3d$-mirror variety $X^{!}$ and a function $\Phi(x,z)$, called {\it superpotential} of $X^{!}$. One of the physically inspired predictions of $3d$-mirror symmetry is that the vertex functions of $X$ then can be represented as
$$
\Ver(z) = \int_{\gamma}\, \Phi(x,z)\, dx
$$
for an appropriate  choice of a multidimensional contour $\gamma$. In this section we give a mathematically precise statement of this construction for the case $X=T^{*}\Gr(k,n)$.

\subsection{} 
Assume that $n\geq 2k$. To a pair $(k,n)$ we associate an $A_{n-1}$ framed quiver as in Fig.\ref{quivpic}. This quiver only has non-trivial one-dimensional framings  at vertices $k$ and $n-k$ (which are represented by the squares in the figure). We define the dimension vector by the formula:
$$
\bv_i=\left\{\begin{array}{ll}
i, & i< k,\\
k, & k \leq i \leq n-k,\\
n-i, & n-k < i,
\end{array}\right. 
$$
Let $X^{!}$ be the Nakajima's quiver variety associated to these data \cite{Nak94, Nak98, MO19}. 
It is known that $X^{!}$ is a $3D$-mirror of $X=T^{*}\Gr(k,n)$, which means that the corresponding vertex functions of $X$ and $X^{!}$ coincide \cite{Din20}. Alternative (but equivalent) definition of $3D$-mirror symmetry requires coincidence of elliptic stable envelope classes for $X$ and $X^{!}$, we refer to \cite{RSVZ19,RSVZ21} for this approach.

\subsection{} 
To a vertex $i$ with dimension $\bv_i$ in a quiver we associate a collection of variables $x_{i,j}$, $j=1,\dots, \bv_i$. In algebraic topology, these variables can be thought of as the Chern roots of the $i$-th tautological bundle over the corresponding quiver variety. 
To  a framing vertex with dimension $\bw_i$ we associate a collection of variables $z_{i,j}$, $j=1,\dots, \bw_i$. The superpotential of a quiver variety is then read of its quiver using the procedure:

\begin{itemize}
    \item To an arrow from a vertex $j$ to a vertex $i$ we associate a factor
    \bean \label{arrowcontr}
    \prod_{a=1}^{\bv_i}   \prod_{b=1}^{\bv_j}\, (x_{i,a}-x_{j,b})^{-\om} 
    \eean 
    \item To a vertex $m$ of the quiver we associate a factor
    \bean
     \label{verterm}
    \prod\limits_{1\leq i<j \leq \bv_m}  (x_{m,i}-x_{m,j})^{2\om}
    \eean
    \item To a vertex $m$ of the quiver we associate a factor:
    \bean \label{detterms}
    \Big( \prod\limits_{j=1}^{\bv_m}\, x_{m,j} \Big)^{-1+\om}
    \eean
\end{itemize}

\noindent 
For the quiver in Fig.\ref{quivpic}, representing the mirror variety $X^{!}$,
 these rules give the following superpotential:
\bean 
\label{superpot}
\Phi(x,z)
&=&
\Big(\prod\limits_{i=1}^{n-1} \prod\limits_{j=1}^{\bv_i}\, x_{i,j} \Big)^{-1+\om} 
\Big(\prod\limits_{m=1}^{\bv_m}\, \prod\limits_{1\leq i<j \leq \bv_m}  (x_{m,j}-x_{m,i})\Big)^{2\om}
\\
\notag
&\times&
\Big(\prod\limits_{i=1}^{n-2} \prod\limits_{a=1}^{\bv_i} \prod\limits_{b=1}^{\bv_{i+1}} (x_{i,a}-x_{i+1,b})\Big)^{-\om} 
\Big(\prod\limits_{i=1}^{k} (z_{k,1}-x_{k,i}) (z_{n-k,1}- x_{n-k,i})\Big)^{-\om} .
\eean

The superpotentials constructed in this way are called the {\it master functions} in the theory of integral representations of the trigonometric Knizhnik-Zamolodchikov equations.  In particular, (\ref{superpot}) corresponds to
the KZ equations associated with the weight subspace of weight [1,\dots,1] in the tensor product the $k$-th and $(n-k)$-th fundamental representations of $\frak{gl}_n$, see \cite{SV91,MV02}. 

\begin{figure}[h!]
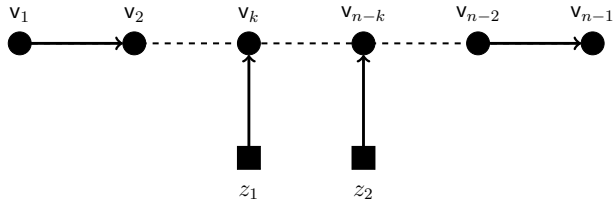

$$
\resizebox{0.5\textwidth}{!}{%
\usebox{\Ipm}}
$$
\caption{The quiver description of the mirror variety $X^{!} \label{quivpic} $} 
\end{figure}

\subsection{}
The total number of variables $x_{i,j}$ equals $\dim \Gr(k,n) = k(n-k)$. Of interest to us are the functions arising as integrals of the superpotential (\ref{superpot}) over 
real $k(n-k)$-dimensional cycles $\gamma \subset \matC^{k(n-k)}$. For such an integral to be well defined, the  superpotential must
have a single-valued branch on $\gamma$. 

For a small real number $0<\epsilon\ll 1$ let us define
$
\epsilon_{i,j}= (|i-k|+2 j-1) \epsilon. 
$
We have an ordering on the pairs $(i,j)$ corresponding to:
\bean
 \label{ordere}
\epsilon_{k,1}< \epsilon_{k-1,1} =  \epsilon_{k+1,1} < \dots  <\epsilon_{n-k,k}
\eean 
ranging from $\epsilon_{k,1}=\epsilon$ to $\epsilon_{n-k,k}=(n-1)\epsilon$. Define the torus $\gamma_{k,n} \subset \matC^{\dim X-1}$ by the system of equations
$
|x_{i,j}|= \epsilon_{i,j}
$
where $i, j$ run through all possible values.

\begin{prop}

 \label{branchprop}

 Assume that $|z_{k,1}|<\epsilon$ and $(n-1) \epsilon <|z_{n-k,1}|$, then  the superpotential (\ref{superpot}) has 
a single-valued branch on the torus $\gamma_{k,n}$, which is distinguished in the proof and which 
will be used in the paper.

\end{prop}

\begin{proof}
 Let us denote 
 $$
 L(x_{i,a},x_{j,b})=\left\{\begin{array}{ll}
 (1-x_{i,a}/x_{j,b})^{-\om}, & \epsilon_{i,a} < \epsilon_{j,b}, \  i\ne j\\
(x_{j,b}/x_{i,a}-1)^{-\om}, & \epsilon_{i,a} > \epsilon_{j,b},  \   i\ne j.
 \end{array}\right.
 $$
Each of these 
ratios $x_{i,a}/x_{j,b}$,
$x_{j,b}/x_{i,a}$ restricted to $\gamma_{k,n}$
has absolute value less than 1.
We replace  $ (1-x_{i,a}/x_{j,b})^{-\om}$ on $\gamma_{k,n}$
with 
$\sum_{m=0}^\infty\binom{-\om}{m}(-x_{i,a}/x_{j,b})^m$ and replace
\\
$(x_{j,b}/x_{i,a}-1)^{-\om}$ with 
$e^{-\pi \sqrt{-1}\om} \sum_{m=0}^\infty\binom{-\om}{m}(-x_{j,b}/x_{i,a})^m$.

Next, we denote 
$L(z_{k,1},x_{k,a})= (1-z_{k,1}/x_{k,a})^{-\om}$ and 
$L(z_{n-k,1},x_{k,a})=
\\
 (1-x_{n-k,a}/z_{n-k,1})^{-\om }$.
On  $\gamma_{k,n}$ we have $|x_{k,i}|\geq \epsilon$, and $|x_{n-k,i}|\leq |x_{n-k,k}|= n \epsilon$, therefore
$|z_{k,i}/x_{k,i}|<1$ and $|x_{n-k,i}/z_{n-k,i}|<1$. 
We replace on $\gamma_{k,n}$ the factor
$(1-z_{k,1}/x_{k,a})^{-\om }$ with 
$\sum_{m=0}^\infty\binom{-\om}{m}(-z_{k,1}/x_{k,a})^m$ and the factor 
$(1-x_{n-k,a}/z_{n-k,1})^{-\om}$ with
\\
$\sum_{m=0}^\infty\binom{-\om}{m}(-x_{n-k,a}/z_{n-k,1})^m$.

Finally, we denote $ L(x_{m,i},x_{m,j})=
 (1-x_{m,i}/x_{m,j})^{2\om }$ for $1\leq i<j\leq \bv_m$.
On $\ga_{k,n}$ we have $|x_{m,i}/x_{m,j}|<1$.
We replace on $\gamma_{k,n}$ the factor
$(1-x_{m,i}/x_{m,j})^{\om}$ with 
$\sum_{m=0}^\infty\binom{2\om}{m}(-x_{m,i}/x_{m,j})^m$.

In these notations we have:

 \bean \label{Lfacts}
 \\
 \notag
 \Phi(x,z)=  \frac{\Big(\prod\limits_{i=1}^{n-1} \prod\limits_{a<b}\Delta(x_{i,a},x_{i,b})\Big)\,
\Big(\prod\limits_{i=1}^{n-2} \prod\limits_{a=1}^{\bv_i} \prod\limits_{b=1}^{\bv_{i+1}} L(x_{i,a},x_{i+1,b})\Big) 
\Big(\prod\limits_{i=1}^{k} L(z_{1},x_{k,i})  L(z_{2},x_{n-k,i})\Big)}{ \prod\limits_{i=1}^{n-1}  \prod\limits_{j=1}^{\bv_i} x_{i,j}  } ,
 \eean

and for each $L$-factor, a single-valued branch is chosen
 by replacing that factor with the corresponding power series. The product of those power series distinguishes a single-valued branch of 
 $\Phi(x,z)$ on $\gamma_{k,n}$.
\end{proof}

We consider the specialization $z_{k,1}=z$ and $z_{n-k,1}=1$, 
which we always assume unless otherwise is stated\footnote{If $n-k=k$ we assume $z_{k,1}=z$ and $z_{k,2}=1$.}. 

\begin{example} For $X=T^*\Gr(2,4)$ we have
\bea
\Phi(x,z) &=& (x_{11}x_{21}x_{22}x_{31})^{-1+\om} (x_{22}-x_{21})^{2\om} 
\\
&\times &
\big( (x_{21}-x_{11}) (x_{31}-x_{21}) (z-x_{21})(1-x_{22})\big)^{-\om}
\\
&\times &
\big( (x_{22}-x_{11}) (x_{31}-x_{22}) (z-x_{22})(1-x_{22})\big)^{-\om}
\\
 &=& (x_{11}x_{21}x_{22}x_{31})^{-1} (1-x_{21}/x_{22})^{2\om} 
\\
&\times &
\big( (x_{21}/x_{11}-1) (1-x_{21}/x_{31}) (z/x_{21}-1)(1-x_{21})\big)^{-\om}
\\
&\times &
\big( (1-x_{11}/x_{22}) (x_{31}/x_{22}-1) (z/x_{22}-1)(1-x_{22})\big)^{-\om}.
\eea

\end{example}

From the previous proposition, the integral of $\Phi(x,z)$ over $\gamma_{k,n}$ is an analytic function of $z$ in the disc $|z|<\epsilon$. This function is represented by a power series in $z$ with complex coefficients, which has the following form:

\begin{thm} 
\label{verttheorem}
The vertex function of $X=T^{*}\Gr(k,n)$ with the equivariant parameters specialized at (\ref{specializ})  has the following 
integral representation
\bean 
\label{vertexInt}
\Ver(z) = \dfrac{\alpha}{(2 \pi \sqrt{-1})^{k(n-k)}} \oint\limits_{\gamma_{k,n}}\, \Phi(x,z) \, \bigwedge\limits_{i,j} dx_{i,j} 
\eean 
where $\Phi(x,z)$ is the branch of superpotential function (\ref{superpot})
 on the torus $\gamma_{k,n}$ chosen in Proposition \ref{branchprop}, and $\alpha=e^{\pi\sqrt{-1}N\om}$ 
 is a normalization constant  where $N$ is the number  of factors 
in \eqref{Lfacts} having the form $(x_{j,b}/x_{i,a}-1)^{-\om}$. 
\end{thm}

 The proof of the theorem found by the authors for arbitrary $k, n$ is
based on the papers \cite{SV91,MV02} , where the integral representations for solutions of
KZ equations were obtained.  Since it  does not pertain directly to the result of the present paper, we give a proof for arbitrary $k$ and $n$ in a separate publication \cite{SmV}. The example below gives a proof of this theorem for $k=1$. 

\begin{example}

 \label{hypergeom_exam}

The case of $k=1$ corresponds to $X=T^{*}\matP^{n-1}$. In this case,
\bea
\Phi(x,z) 
&=& 
(x_{1,1}\dots x_{n-1,1})^{-1+\om} \big((x_{1,1}-x_{2,1})\dots (x_{n-2,1}-x_{n-1,1})(z-x_{1,1})(1-x_{n-1,1})\big)^{-\om}
\\
&=&  (x_{1,1}\dots x_{n-1,1})^{-1}
\Big(\Big(\frac{x_{1,1}}{x_{2,1}}-1\Big)\dots \Big(\frac{x_{n-2,1}}{x_{n-1,1}}-1\Big)
  \Big(\frac{z}{x_{1,1}}-1\Big)  \Big(1-x_{n-1,1}\Big)\Big)^{-\om}.
\eea
The integral (\ref{vertexInt}) takes the form:
\bean 
\label{tpninteg}
\dfrac{\alpha}{(2 \pi \sqrt{-1})^{n-1}}\oint_{\gamma}\, \Phi(x,z) \, dx_{1,1} \wedge \dots \wedge dx_{n-1,1}
\eean
where $\gamma \in \matC^{n-1}$ is the torus given by equations
$
|x_{i,1}|=i \epsilon,  \ i=1,\dots, n-1.
$
We replace each binomial by  a power series using the formulas:
\bea
(1-a)^{-\om}=\sum\limits_{m=0}^{\infty}\, \binom{-\om}{m} (-a)^m,
\qquad
(a-1)^{-\om}=e^{-\pi \sqrt{-1}/2} \sum\limits_{m=0}^{\infty}\, \binom{-\om}{m} (-a)^m.
\eea
Using these expansions and noting that $\alpha=e^{(n-1) \pi \sqrt{-1}\om}$ for  (\ref{tpninteg}),
 we obtain that the integral equals
\begin{small}
$$
\dfrac{1}{(2 \pi \sqrt{-1})^{n-1}} \sum\limits_{m_1,\dots,m_n\geq 0}\,  (-1)^{m_1+\dots+m_n}  \binom{-\om}{m_1}
\dots \binom{-\om}{m_n} \oint\limits_{\gamma} \dfrac{\Big(\frac{z}{x_{1,1}}\Big)^{m_1}\Big(\frac{x_{1,1}}{x_{2,1}}\Big)^{m_2}\dots x_{n-1,1}^{m_n}}{x_{1,1} \dots x_{n-1,1}}
$$
\end{small}
The last integral is computed by evaluating the residues consequently from $x_{1,1}=0$ to $x_{n-1,1}=0$. For instance, the residue at $x_{1,1}=0$ is non zero only if $m_1=m_2$. The residue at $x_{2,1}=0$ in non-zero only if $m_2=m_3$ and so on. Thus, $m_1=m_2=\dots = m_{n}$ and the result is the power series
$$
\sum\limits_{d=0}^{\infty}\, (-1)^{n d} \binom{-\om}{d}^n\! z^{d}.
$$
The last sum is the well-known expansion of the generalized hypergeometric function:
\bean
 \label{hypergeom_ex}
{}_{n-1} F_{n}\big(\om,\dots, \om; 1,\dots,1; z\big).
\eean
This function coincides with the vertex function $\Ver(z)$ for $T^{*}\matP^{n-1}$
computed in Section 6.2  of \cite{AO16} 
(which is done there for generic values of parameters). We note also that for $n=1$, we obtain the function from Example \ref{examtp1}. 
\end{example}

\begin{rem} 
\label{intresidues}
For general values of $k,n$,\, the integral (\ref{vertexInt}) can be evaluated as in the previous example, i.e., by expanding the integrand into power series in $x_{i,j}$ and $z$ as in Proposition \ref{branchprop} and then computing the residues at $x_{i,j}=0$.  By definition of the torus
 $\gamma_{k,n}$, the residues are to be computed in order compatible with (\ref{ordere}). 

\end{rem}

\section{$p$-adic approximations of vertex functions}

\subsection{Polynomial superpotentials} 

 In the remainder of the paper we assume that 
\bean
\label{omrq}
\om = r/q, \qquad r,\,q\,\on{positive \,integers}, \qquad r/q\leq 1/2.
\eean
Let $p$ be an odd  prime number of the form
\bean
\label{pq form}
p=\ell q+1, \qquad\ell\,\, \on{a\, \,positive\,\,integer}.
\eean

It is useful to rearrange the factors of the superpotential.
We have
\bean
 \label{Phi}
 \Phi(x,z) \,=\, \Delta(x) \,\bar \Phi(x,z)
 \eean
 where
 \bean 
 \label{deltfact}
\Delta(x)  =\prod\limits_{m=1}^{\bv_m}\, \prod\limits_{1\leq i<j \leq \bv_m}  (x_{m,j}-x_{m,i})
\eean
and
\footnote{In the case $k=n-k$ the pair $(z_{k,1},z_{n-k,1})$ becomes $(z_{k,1}, z_{k,2})$, which we assume throughout.} 
\bean
\label{bar Phi}
\phantom{aaa}
\bar\Phi(x,z)
&=&\Big(\prod\limits_{i=1}^{n-1} \prod\limits_{j=1}^{\bv_i}\, x_{i,j} \Big)^{-1+r/q} 
\Big(\prod\limits_{m=1}^{\bv_m}\, \prod\limits_{1\leq i<j \leq \bv_m}  (x_{m,j}-x_{m,i})\Big)^{-1+2r/q}
\\
&\times&
\Big(\prod\limits_{i=1}^{n-2} \prod\limits_{a=1}^{\bv_i} \prod\limits_{b=1}^{\bv_{i+1}} (x_{i,a}-x_{i+1,b})\Big) ^{-r/q}
\Big(\prod\limits_{i=1}^{k} (z_{k,1}-x_{k,i}) (z_{n-k,1}- x_{n-k,i})\Big)^{-r/q} .
\notag
\eean

For any integer $s\geq 1$, we define the following polynomial approximation of the superpotential function (\ref{superpot}):
\bean
 \label{polsuper}
\Mas_{s}(x,z)= \Delta(x)  \bar\Phi(x,z)^{1-p^s}.
\eean
Notice that  the power $1-p^s$  
approaches $1$ in the $p$-adic norm for large $s$. 
We denote
\bean
\label{bar Phi}
 \bar \Phi_s (x,z)  = \bar \Phi (x,z)^{1-p^s}.
\eean

\begin{lem}
\label{lem 4.1}

${}$

\begin{itemize}
\item $ \bar \Phi_s (x,z)$ is a polynomial.

\item For any $a=1,\dots,n-1$, the polynomial  $\bar \Phi_s (x,z)$
is symmetric with respect to permutation of the variables $x_{a,1},\dots,x_{a,\bv_a}$.

\end{itemize}

\end{lem}

\begin{proof}  We have
\bean
\label{bps}
&&
\\
\notag
\bar\Phi_s(x,z)
&=&\Big(\prod\limits_{i=1}^{n-1} \prod\limits_{j=1}^{\bv_i}\, x_{i,j} \Big)^{(p^s-1)(q-r)/q} 
\Big(\prod\limits_{m=1}^{\bv_m}\, \prod\limits_{1\leq i<j \leq \bv_m}  (x_{m,j}-x_{m,i})\Big)^{(p^s-1)(q-2r)/q}
\\
&\times&
\Big(\prod\limits_{i=1}^{n-2} \prod\limits_{a=1}^{\bv_i} \prod\limits_{b=1}^{\bv_{i+1}} (x_{i,a}-x_{i+1,b})\Big) ^{(p^s-1)r/q}
\Big(\prod\limits_{i=1}^{k} (z_{k,1}-x_{k,i}) (z_{n-k,1}- x_{n-k,i})\Big)^{(p^s-1)r/q}.
\notag
\eean
Notice that $(p^s-1)(q-r)/q, \,(p^s-1)(q-2r)/q,\,(p^s-1)r/q$ are positive integers by
assumptions \eqref{omrq} and \eqref{pq form}, moreover, the integer
 $(p^s-1)(q-2r)/q$ is even.

The first, third and fourth products are clearly symmetric with respect to permutations
of $x_{a,1},\dots,x_{a,\bv_a}$. The second product is symmetric
since  $(p^s-1)(q-2r)/q$ is even. 
\end{proof}

To keep track of degrees of polynomials in the variables $x_{i,j}$ we will use $(n-1)$-tuples of degree vectors
$u=(u^{(1)},\dots,u^{(n-1)})$ with $u^{(i)}=(u^{(i)}_1,\dots, u^{(i)}_{\bv_i}) \in \mathbb{N}^{\bv_i}$. With this notation 
$x^{u}$ denotes the monomial
$$
x^{u}= \prod\limits_{i=1}^{n-1} \prod\limits_{j=1}^{\bv_i}\, x_{i,j}^{u^{(i)}_{j}}.
$$

\subsection{The $p$-adic approximations of the vertex function \label{sectionT}} 
Let us define a degree vector $d$ by:
\bean \label{defn}
d^{\,(i)}_{j}  = j
\eean
for  $i=1,\dots,n-1$ and   $j=1,\dots, \bv_{i}$. The following polynomials are the main objects of this paper.

\begin{defn}
 Define the polynomials $\T_s(z) \in \matZ[z]$ by the formula
\bean \label{definitionofT}
\T_s(z) := 
(-1)^{\frac{(p^s-1)r}{q} N} \, 
\mathrm{coeff}_{x^{d p^s -1}} \Big(\Mas_s(x,z) \Big) 
\eean
where $ \mathrm{coeff}_{x^{d p^s -1}}$ denotes the coefficient of the monomial
$
x^{d p^s -1} = \prod\limits_{i=1}^{n-1} \prod\limits_{j=1}^{\bv_i}\, x_{i,j}^{j p^s -1 }
$
in the  polynomial $\Mas_s(x,z)$  and $N$ is the number from Theorem \ref{verttheorem}.  
\end{defn}

A simple degree count shows that 
$\T_s(z)$ has degree
$(p^s-1)kr/q$ in $z$. The prefactor $(-1)^{\frac{(p^s-1)r}{q} N} $ in the definition is 
introduced to fix the constant term of this polynomial,
 $\T_s(0)=1$, as explained in the lemma below:
\begin{lem}
Let $\Mas_s(x,0)$ be the polynomial superpotential  with $z=0$, then
 $$
 \mathrm{coeff}_{x^{d p^s -1}} \Big(\Mas_s(x,0) \Big)= (- 1)^{\frac{(p^s-1)r}{q} N}.
 $$
\end{lem}
\begin{proof}
The proof is by direct computation of coefficients of $x^{j p^s-1}_{i,j}$ in the order on pairs $(i,j)$ given by (\ref{ordere}). In the first step, we need to compute the coefficients of 
$x_{k,1}^{p^s-1}$ in 
$$
\Phi_s(x,0) = m_1(x) \,b_1(x)
$$
where we separated the part $m_1(x)$ given by the product of monomials and part $b_1(x)$ given by the product of binomial factors of the form $(a-b)^c$. The variable  $x_{k,1}$ enters the monomial part $m_1(x)$ as a factor $x_{k,1}^{p^s-1}$.  This is already the full degree we need, and therefore the binomial part can only contribute a constant factor in $x_{k,1}$, i.e., we have:
$$
 \mathrm{coeff}_{x^{{p^s-1}}_{k,1}} (\Phi_s(x,0))= \left.m_1(x)\right|_{x_{k,1}=1}  \left.b_1(x)\right|_{x_{k,1}=0} 
$$
The effect of substituting $x_{k,1}=0$ into the binomial part
$b_1(x)$ is that the binomial factors containing $x_{k,1}$, which are of the form $(s-x_{k,1})^c$ or $(x_{k,1}-s)^{c}$
turn into the monomials $(s)^c$ or  $(-s)^{c}$ respectively. Again, separating all factors into a product of monomials and binomials we obtain
\bean 
\label{stepp2}
\left.m_1(x)\right|_{x_{k,1}=1}  \left.b_1(x)\right|_{x_{k,1}=0} = m_2(x) b_2 (x)
\eean
In the second step, we need to compute the coefficient of 
$x^{p^s-1}_{k-1,1}$ in (\ref{stepp2}). A simple computation shows that the variable $x_{k-1,1}$ enters the monomial part as $x_{k-1,1}^{p^s-1}$, i.e., it again has the full degree and the whole process is repeating. We claim that the same is true for any step in the sequence.

Indeed, assume that after $l-1$ steps we arrived at $m_l(x) b_l(x)$, where $m_l(x)$ and $b_l(x)$ denote the monomial and the binomial parts as before. Assume that for $l$-th step we need to compute the coefficient of  $x_{a,b}^{b p^s-1}$ in this expression. 

Let us compute the degree of $x_{a,b}$ in the monomial part
$m_l(x)$: there are $(2 b-1)$ factors in $\Phi_s(x,z)$ of the form $(x_{a-1,d}-x_{a,b})^{(p^s-1) r/q}$ or $(x_{a,b}-x_{a+1,d})^{(p^s-1) r/q}$ for which we have already substituted $x_{a-1,d}=0$, $x_{a+1,d}$ in the previous steps (these are the factors with  $(a-1,d)<(a,b)$ or $(a+1,d)<(a,b)$ in the order (\ref{ordere})). Each of these factors contributes $\pm x_{a,b}^{(p^s-1) r/q}$ to $m_l(x)$. 
Similarly, there are exactly $b-1$ factors of the form
$(x_{a,b}-x_{a,b'})^{(p^s-1)(q-2 r)/q+1}$ in which we substituted $x_{a,b'}=0$ in the previous steps (these are factors with $b'<b$). Each of these factors contributes the monomial $x_{a,b}^{(p^s-1)(q-2 r)/q+1}$ in $m_l(x)$. Finally, there is a factor $x_{a,b}^{(p^s-1)(q-r)/q}$, which was already in $m_1(x)$ in the very first step.  In total, we obtain that the degree of $x_{a,b}$ in the monomial $m_l(x)$ is
$$
(2b-1) \frac{(p^s-1) r}{q} + (b-1)\Big( \dfrac{(p^s-1)(q-2 r)}{q} +1 \Big) + \frac{(p^s-1)(q-r)}{q}  = b p^s-1
$$
which gives the full degree for the monomial $x_{a,b}$ and we obtain that
$$
 \mathrm{coeff}_{x^{{b p^s-1}}_{a,b}} (m_l(x) b_l(x)) =\left.m_{l}(x)\right|_{x_{a,b}=1} \left.b_l(x)\right|_{x_{a,b}=0} = m_{l+1}(x) b_{l+1}(x)
$$
where in the last step we again separated the monomial part $m_{l+1}(x)$ and the binomial part $b_{l+1}(x)$. 

Repeating these calculations, after $(n-k)k$ steps we arrive at the last variable $x_{n-k,k}$. Clearly the binomial part must be trivial $b_{(n-k)k}=1$ and the monomial part has full degree as we proved before, i.e., $m_{k (n-k)}(x) = \pm x^{k p^s-1}_{n-k,k}$. Thus we conclude
$$
\mathrm{coeff}_{x^{d p^s -1}} \Big(\Mas_s(x,0) \Big)=\pm 1.
$$
To commute the sign, we note that the monomial parts are multiplied by $(-1)^{(p^s-1) r/q}$ whenever we substitute $x_{a,b}=0$  to one of the factors in the binomial parts
which is of the form $(x_{a,b}-x_{a+1,c})^{(p^s-1) r/q}$ where $(a,b)<(a+1,c)$ in the order (\ref{ordere}). 
The number of such pairs is exactly what we denoted by $N$ in the Theorem \ref{verttheorem}. 
\end{proof}

The polynomials $\T_s(z)$ can be viewed as {\it $p$-adic approximations of the vertex functions} since $\T_s(z) \to \Ver(z)$ as $s\to \infty$ in the following sense.
\begin{thm} 
\label{TtoVthm}
 Consider the expansions: 
$$
\Ver(z)=\sum\limits_{m=0}^{\infty}\, c_m\, z^{m}, \ \ \ 
\T_s(z)=\sum\limits_{m=0}^{\deg \T_s(z)}\, c_{s,m}\, z^{m}.
$$
Then for any $m\geq 0$, the sequence of integers $c_{s,m}$,  converges in the $p$-adic norm and
\bean \label{limitpart}
\lim\limits_{s\to \infty}\, c_{s,m} = c_m \,.
\eean

\end{thm}

\noindent
Before we proceed to the proof, let us consider an example.

\begin{example}
Let  $k=1$, $n=2$, $\om =1/2$. In this case
$\T_s(z)$ is the coefficient of $x^{p^s-1}$ in the polynomial
$$
\Mas_s(x,z)=(-1)^{\frac{p^s-1}{2}} \Big(  x(x-1)(x-z)  \Big)^{\frac{p^s-1}{2}}.
$$
An elementary computation gives
$$
\T_s(z)=\sum\limits_{d=0}^{\frac{p^s-1}{2}}\, \binom{\frac{p^s-1}{2}}{d}^2\, z^d.
$$
In the $p$-adic norm we have
$$
\lim\limits_{s\to \infty}\, \binom{\frac{p^s-1}{2}}{d} = \binom{\frac{-1}{2}}{d},
$$
and in the limit $s\to \infty$ we arrive at the hypergeometric series (\ref{hypertp1}) with $\om=1/2$.

\end{example}

\begin{proof}
The proof, essentially, generalizes the computation in the previous example to the case of several variables. 
As in Remark \ref{intresidues}, the coefficients of the $z$-power series given by the integral (\ref{vertexInt}) can be evaluated by expanding 
$\Phi(x,z)$
into power series in $x_{i,j}$ and $z$ and computing residues at $x_{i,j}=0$ in a certain order. 
Using the Newton binomial theorem we find that the coefficient of a
$z^{d}$ is given by a finite sum of products of the binomial coefficients of the form
$\binom{2r/q}{i}$
or
$\binom{-r/q}{i}$
with certain degree constrains on the indices $i$. These constrains come from computing the residues, i.e., the sum of the degrees $i$ must give the monomial $x^{-1}=\prod_{i,j}\, x_{i,j}^{-1}$.

Similarly, the coefficients  of the polynomial $\T_s(z)$ are  computed 
by expanding $\Phi_{s}(x,z)$
into power series in $z$ using  the Newton binomial theorem
and then computing the coefficient of $x^{d p^s -1}$. We find that
the coefficient of a $z^d$  is a sum of products of binomial coefficients of the form
$\binom{(p^s-1)r/q}{i}$
or
$\binom{1+(p^s-1)(q-2r)/q}{i}$
with certain constrains on the indices~$i$.   
Namely, the sum of degrees $i$ corresponds to the monomial $x^{d p^s-1}=\prod_{i,j}\, x_{i,j}^{j p^s-1}$.

As $s\to \infty$, the degree constrains to compute the coefficient of a $z^d$ coincide,
 and in the $p$-adic norm
we have  $\binom{(p^s-1)r/q}{i}\to \binom{-r/q}{i}$
and $\binom{1+(p^s-1)(q-2r)/q}{i}\to \binom{2r/q}{i}$,
 which gives (\ref{limitpart}). 
\end{proof}

\section{Dwork-type congruences for $\T_s(z)$  }

\subsection{Dwork-type congruences}
The goal of this section is to prove our main theorem:

\begin{thm} 
\label{dworkth}

 The polynomials $\T_s(z)$ satisfy the Dwork-type congruences:
\bean
 \label{dworkid}
\dfrac{\T_{s+1}(z)}{\T_s(z^p)}=\dfrac{\T_{s}(z)}{\T_{s-1}(z^p)} \mod p^s
\eean
with $s=1,2,\dots$ and $\T_{0}(z)=1$. 
\end{thm}

Before we proceed with the proofs, we discuss several consequences.

\begin{cor} \label{corrprod}
 For $m=0,\dots, s-1$, the polynomials $\T_s(z)$ satisfy the following congruences:
\bean \label{recur}
\T_{s}(z)=\dfrac{\T_{s-m}(z)\T_{s-m}(z^p)\dots \T_{s-m}(z^{p^m})}{\T_{s-m-1}(z^p)\dots \T_{s-m-1}(z^{p^m})} \mod p^{s-m}
\eean
In particular, for $m=s-1$, we have
$$
\T_{s}(z)=\T_1(z) \T_1(z^p) \dots \T_1(z^{p^{s-1}}) \mod p
$$

\end{cor}

\begin{proof}
By substituting $s\to s-1$ we rewrite relation (\ref{dworkid}) in the form
 \bean
  \label{step0}
 \T_s(z)=\dfrac{\T_{s-1}(z) \T_{s-1}(z^p)}{\T_{s-2}(z^p)} \mod p^{s-1}
 \eean 
 which gives (\ref{recur}) for $m=1$.  For the second step, we use this relation to substitute the factors $\T_{s-1}(z)$ in the numerator of (\ref{step0}) by
 $$
 \T_{s-1}(z)=\dfrac{\T_{s-2}(z) \T_{s-2}(z^p)}{\T_{s-3}(z^p)} \mod p^{s-2}
 $$
which gives
 $$
 \T_s(z)=\dfrac{\T_{s-2}(z) \T_{s-2}(z^p) \T_{s-2}(z^{p^2})}{\T_{s-3}(z^p) \T_{s-3}(z^{p^2})} \mod p^{s-2}
 $$
 i.e., (\ref{recur}) for $m=2$. Continuing by induction, after $m$ steps we arrive at~(\ref{recur}). 
\end{proof}

\begin{thm} \label{infprodrepr}

For $a\geq 1$, the vertex function of $X=T^{*}\Gr(k,n)$ has the following infinite product presentation
modulo $p^a$, 
 $$
\Ver(z)=\prod\limits_{i=0}^{\infty}\, \dfrac{\T_a(z^{p^{i}})}{\T_{a-1}(z^{p^{i+1}})} \mod p^a
 $$  
which means that the coefficients of the Taylor series
of both sides at $z=0$ are equal modulo~$p^a$. In particular, for $a=1$ we obtain
$$
\Ver(z)=\prod\limits_{i=0}^{\infty}\,\T_1(z^{p^{i}}) \mod p.
$$
 
\end{thm}

\begin{proof}
 In (\ref{recur}) we consider the limit as $s\to \infty$, $m\to \infty$ such that $m-s=a$ is fixed.
  By Theorem \ref{TtoVthm} this limit converges to the vertex function $\Ver(z)$.  
\end{proof}

\subsection{Ghosts, cf. \cite{VZ21,Var22b}} 
Define the polynomials $\Ghost_s(x,z)$, $s\geq 0$, recursively:
$$
\Ghost_0(x,z)=\Mas_1(x,z)
$$
and 
\bean
 \label{inddefgh}
\Ghost_s(x,z)= \Mas_{s+1}(x,z) - \sum\limits_{j=1}^{s}\, \Ghost_{j-1}(x,z)\, \overline{\Mas}_{s-j+1}(x^{p^j}, z^{p^j}),
\qquad s\geq 1.
\eean
For example,
$$
\Ghost_1(x,z)=\Mas_{2}(x,z) - \Mas_{1}(x,z) \overline{\Mas}_{1}(x^p,z^p)
$$
Note that 
$$
\Mas_{2}(x,z) = \Delta(x) \overline{\Mas}_{1}(x,z)^{1+p},  \ \ \Mas_{1}(x,z) = \Delta(x) \overline{\Mas}_{1}(x,z)
$$
Therefore
$$
\Ghost_1(x,z)= \Delta(x)\overline{\Mas}_{1}(x,z) (\overline{\Mas}_{1}(x,z)^p-\overline{\Mas}_{1}(x^p,z^p))
$$
It is clear from the last formula that
$$
\Ghost_1(x,z) = 0 \mod p.
$$
One can  easily generalize this property using induction on $s$:

\begin{lem}
 The polynomials $\Ghost_s(x,z)$ satisfy $\Ghost_{s}(x,z) = 0 \mod p^s$.
 \qed
\end{lem}

Let us define the {\it ghost polynomials} $\G_s(z)\in \mathbb{Z}[z]$ by 
$$
\G_s(z)=\textrm{coeff}_{x^{d p^s-1}}\Big( \Ghost_{s-1}(z)  \Big),
$$
where $x^{d p^s-1}$ is the same monomial as in   \eqref{definitionofT}. The above lemma implies that
\bean \label{ghvanish}
\G_s(z) =0  \mod p^{s-1}.
\eean

\subsection{Ghosts expansions of $\T_s(z)$}

The group $\frak{S}_{k,n} =\frak{S}_{\bv_1} \times \frak{S}_{\bv_2}\times \dots \times  \frak{S}_{\bv_{n-1}}$ acts naturally on the set of variables $x_{i,j}$ (the symmetric group $\frak{S}_{\bv_i}$ acts by permutations of variables 
$x_{i,1},\dots, x_{i,\bv_{i}}$). 
By Lemma \ref{lem 4.1}, the polynomial  $\overline{\Mas}_{s}(x,z)$ is invariant under action of $\frak{S}_{k,n}$ 
while $\Mas_{s}(x,z)$ is skew-symmetric:
\bean \label{symtry}
\Mas_{s}(\sigma(x),z)=(-1)^{\sigma} \Mas_{s}(x,z), 
\qquad 
\overline{\Mas}_{s}(\sigma(x),z)=\overline{\Mas}_{s}(x,z)
\eean
where $(-1)^{\sigma}$ is the sign of a permutation $\sigma \in \frak{S}_{k,n}$. 
\vsk.2>

For this section, let us redefine:
\bean \label{redefT}
\T_s(z) =  \mathrm{coeff}_{x^{d p^s -1}} \Big(\Mas_s(x,z) \Big) 
\eean
which differs from our definition (\ref{definitionofT}) by the factor
 $(-1)^{\frac{(p^s-1)r}{q} N}$. It will be easier to prove  Theorem 
 \ref{dworkth} for this polynomial, and then show that the result does not depend on this sign.

\begin{lem} \label{tfrommbar}
In this notation we have the following equality:
\bean \label{Tdef}
\T_s(z) = \sum\limits_{\sigma \in \frak{S}_{k,n}}\, (-1)^\sigma\, \mathrm{coeff}_{x^{d p^s - \sigma(d)}} ( \overline{\Mas}_{s}(x,z) ) 
\eean
where $d$ is given by (\ref{defn}). 
\end{lem}
\begin{proof}
Expanding the product $\Delta(x)$ we find
 $$
 \Delta(x)=\sum\limits_{\sigma \in \frak{S}_{k,n}}\, (-1)^{\sigma} \sigma\Big(\prod\limits_{m} \, x_{m,1}^{0} x_{m,2}^{1} \dots x_{m,\bv_m}^{\bv_m-1} \Big) = \sum\limits_{\sigma \in \frak{S}_{k,n}}\, (-1)^{\sigma} x^{\sigma(d)-1}
 $$
By definition  $\Mas_{s}(x,z)=\Delta(x) \overline{\Mas}_{s}(x,z)$. Combining this with definition (\ref{definitionofT}) we arrive at (\ref{Tdef}).
\end{proof}

\vspace{3mm}

\begin{lem} 
\label{lemmaforT}
 Let $u \in \mathbb{Z}^{D}$ be a vector of multi-degrees. 
The coefficient of $x^{u p^s -1}$ in the polynomial $\Mas_s(x)$ equals zero unless  $u$ is of the form
$$
u = \sigma (d), \ \ \sigma \in \frak{S}_{k,n}\,,
$$
where $d$ is defined by (\ref{defn}).   If $u = \sigma (d)$, then
\bean
\mathrm{coeff}_{x^{u p^s -1}} (\Mas_s(x,z)) = 
(-1)^{\sigma} \T_s(z).
\eean 
\end{lem}

\begin{proof}
Recall that $\Mas_s(x)$ is skew-symmetric (\ref{symtry}).  Thus, the coefficient of $x^{u p^s -1}$ in $\Mas_s(x)$ 
may be nonzero only if the list  
\bean \label{mlist}
u^{(m)}=(u^{(m)}_1, \dots, u^{(m)}_{\bv_m})
\eean
consists of  pairwise distinct integers for all $m$. 

It follows from formula \eqref{bps} that
\bean \label{degM}
\deg_{x_{m,i}}(\Mas_s(x,z)) < (\bv_m+1) p^s-1.
\eean

\bigskip

Thus, a non-zero monomials of the form $x_{m,i}^{u^{(m)}_{i} p^s-1}$ can appear in the polynomial $\Mas_s(x,z)$ only for $u^{(m)}_{i}$ satisfying $1 \leq u^{(m)}_i \leq \bv_m$. Since the elements in the list (\ref{mlist}) are pairwise distinct and satisfy the the bound  $1 \leq u^{(m)}_i \leq \bv_m$, it must be of the form
$
u^{(m)}=\sigma((1,2,\dots, \bv_m))
$
for some permutation $\sigma \in \frak{S}_{\bv_m}$. This proves the first statement of 
the lemma. Now, assume that $u=\sigma(d)$, then from the skew-symmetry of $\Mas_s(x,z)$ we have
$$
\mathrm{coeff}_{x^{\sigma(d) p^s -1}} (\Mas_s(x,z)) =  (-1)^{\sigma} \mathrm{coeff}_{x^{d p^s -1}} (\Mas_s(x,z))  = (-1)^{\sigma} \T_s(z).
$$
where the last equality is the definition of $\T_s(z)$.
\end{proof}

\noindent
The same result holds for ghosts:

\begin{lem} \label{lemmaForG}
Let $u \in \mathbb{Z}^{D}$ then the coefficients of $x^{u p^s -1}$ in the polynomial $\Ghost_{s-1}(x,z)$ is equal to zero unless 
the degree vector $u$ is of the form
$$
u = \sigma (d), \qquad \sigma \in \frak{S}_{k,n}\,,
$$
where $d$ is defined by (\ref{defn}). If $u = \sigma (d)$, then
\bean \label{sksymG}
\mathrm{coeff}_{x^{\sigma(d) p^s -1}} (\Ghost_{s-1}(x,z)) = (-1)^{\sigma} \G_s(z).
\eean
\end{lem}

\begin{proof}
From the inductive definition  (\ref{inddefgh}) it is clear that they are skew-symmetric:
$$
\Ghost_{s-1}(\sigma(x),z)=(-1)^{\sigma } \Ghost_{s-1}(x,z)
$$
and the degree of a 
variable $x_{m,i}$ in the polynomial $\Ghost_{s-1}(x,z)$ has the same bound as in $\Mas_s(x,z)$.
 From (\ref{degM})  we find
$$
\deg_{x_{m,i}} (\Ghost_{s-1}(x,z)) < (\bv_m+1) p^s-1.
$$
Since this degree bound and the skew-symmetry are the only properties of $\Mas_s(x,z)$ used in the proof of Lemma \ref{lemmaforT}, the same logic applies to $\Ghost_{s-1}(x,z)$.
\end{proof}

\begin{thm}
The polynomial $\T_s(z)$ has the following expansion in ghosts:
\bean \label{Tghexp}
\T_s(z)= \sum\limits_{m=1}^{s}\, \G_{m}(z) \T_{s-m}(z^{p^m}),
\eean
where $\T_0(z)=1$. 
\end{thm}

\begin{proof}
 From definition of ghosts (\ref{inddefgh}) we have:
 $$
 \Mas_s(x,z)=\sum\limits_{m=1}^{s}\, \Ghost_{m-1}(x,z) \overline{\Mas}_{s-m}(x^{p^m},z^{p^m}).
 $$
 By definition,   
 $ \T_s(z)=\mathrm{coeff}_{x^{d p^s -1}} \Big(  \Mas_s(x,z) \Big).$
Thus, it is enough to prove that
$$
\mathrm{coeff}_{x^{d p^s -1}} \Big(\Ghost_{m-1}(x,z) \overline{\Mas}_{s-m}(x^{p^m},z^{p^m}) \Big) = \G_{m}(z) \T_{s-m}(z^{p^m}).
$$
We compute 
$$
\mathrm{coeff}_{x^{d p^s -1}} \Big(\Ghost_{m-1}(x,z) \overline{\Mas}_{s-m}(x^{p^m},z^{p^m}) \Big) = 
$$
$$
\sum\limits_{{{\alpha \in N(\Ghost_{m-1}(x,z)),} \atop 
{\beta \in N(\bar{\Mas}_{s-m}(x,z)),}} \atop 
{\alpha + p^m \beta = d p^s -1}
}\, \mathrm{coeff}_{x^{\alpha}} ( \Ghost_{m-1}(x,z) )\, \mathrm{coeff}_{x^{\beta}} ( \overline{\Mas}_{s-m}(x,z^{p^m}) )
$$
where $N(A(x))$ denotes the  Newton polygon of a polynomial $A(x)$ in variables~$x$.  The condition $\alpha + p^m \beta = d p^s -1$ implies that
$
\alpha+1 = p^m(p^{s-m} d- \beta)
$
i.e. $\alpha+1 \in p^m \mathbb{Z}^{D}$. This means that
$\alpha$ is of the form
$$
\alpha = u p^{m }-1
$$
where $u \in \mathbb{Z}^D$ is some multi-degree vector. By Lemma \ref{lemmaForG} the coefficients $\mathrm{coeff}_{x^{\alpha}} ( \Ghost_{m-1}(x,z) )$ are non-zero only if $u=\sigma(d)$.  We conclude that $\alpha$ and $\beta$ must be of the form
$$
\alpha=\sigma(d) p^m -1, \ \ \ \beta=d p^{s-m} -\sigma(d), \ \ \sigma \in \frak{S}_{k,n} 
$$
Thus, the above sum takes the form:
$$
\sum\limits_{\sigma \in \frak{S}_{k,n} }\, \mathrm{coeff}_{x^{\sigma(d) p^m -1}} ( \Ghost_{m-1}(x,z) )\, \mathrm{coeff}_{x^{d p^{s-m} -\sigma(d)}} ( \overline{\Mas}_{s-m}(x,z^{p^m}) )
$$
Now by (\ref{sksymG}) we have
$$
\mathrm{coeff}_{x^{\sigma(d) p^m -1}} ( \Ghost_{m-1}(x,z) ) = (-1)^{\sigma} \G_m(z)
$$
and the above sum factors
$$
\G_m(z)\, \sum\limits_{\sigma \in \frak{S}_{k,n} }\, (-1) ^{\sigma} \mathrm{coeff}_{x^{d p^{s-m} -\sigma(d)}} ( \overline{\Mas}_{s-m}(x,z^{p^m}) ) = \G_{m}(z) \T_{s-m}(z^{p^m})
$$
where the last equality is by Lemma \ref{tfrommbar}.
\end{proof}

\subsection{Proof of Theorem \ref{dworkth}}

We prove  Theorem \ref{dworkth} by induction on $s$. 
Assume that the theorem is proved for all indices less than $s$, i.e., the identities:
$$
\dfrac{\T_{s-k}(z)}{\T_{s-k-1}(z^p)}=\dfrac{\T_{s-k-1}(z)}{\T_{s-k-2}(z^p)} \mod p^{s-k-1}
$$
hold for all $k=0,\dots,s-2$. Substituting $z\to z^{p^{k}}$ 
into $k$-th identity and multiplying first $m-2$ of them, after telescopic cancellation we obtain:
$$
\dfrac{\T_{s}(z)}{\T_{s-m+1}(z^{p^{m-1}})} = \dfrac{\T_{s-1}(z)}{\T_{s-m}(z^{p^{m-1}})} \mod p^{s-m+1}
$$
By substituting $z\to z^p$ and taking inverses of both sides we obtain:
\bean
 \label{modrel1}
\dfrac{\T_{s-m+1}(z^{p^{m}})}{\T_{s}(z^p)} = \dfrac{\T_{s-m}(z^{p^{m}})}{\T_{s-1}(z^p)} \mod p^{s-m+1}
\eean
Now, using (\ref{Tghexp}) we find:
\bean \label{gss}
\dfrac{\T_{s+1}(z)}{\T_s(z^p)} = \sum_{m=1}^{s+1} \G_m(z) \dfrac{\T_{s-m+1}(z^{p^m})}{\T_s(z^p)}
\eean
and
$$
\dfrac{\T_{s}(z)}{\T_{s-1}(z^p)} = \sum_{m=1}^{s} \G_m(z) \dfrac{\T_{s-m}(z^{p^m})}{\T_{s-1}(z^p)}
$$
Note that the last term in (\ref{gss}) vanishes modulo $p^s$, since $\G_{s+1}(z)=0 \mod p^s$ by (\ref{ghvanish}). Thus we obtain:
$$
\dfrac{\T_{s+1}(z)}{\T_s(z^p)} -\dfrac{\T_{s}(z)}{\T_{s-1}(z^p)} =\sum\limits_{m=1}^{s}\,\G_m(z)
\left(\dfrac{\T_{s-m+1}(z^{p^m})}{\T_s(z^p)} - \dfrac{\T_{s-m}(z^{p^m})}{\T_{s-1}(z^p)} \right) \mod p^s
$$
From (\ref{modrel1}) and from $\G_{m}(z)\equiv 0 \mod p^{m-1}$ we see that each term in the last sum is divisible by $p^{s}$, thus:
$$
\dfrac{\T_{s+1}(z)}{\T_s(z^p)} -\dfrac{\T_{s}(z)}{\T_{s-1}(z^p)} \equiv 0 \mod p^s
$$
Thus, we proved the theorem for the polynomials (\ref{redefT}), which differ from definition (\ref{definitionofT}) by a rescaling $\T_s(z)\to (-1)^{\frac{(p^s-1)r}{q} N} \T_s(z)$. Upon this rescaling, the right-hand side of (\ref{dworkid}) is multiplied by
\bea
\frac{(-1)^{\frac{(p^{s+1}-1)r}{q} N}}{(-1)^{\frac{(p^s-1)r}{q} N}} = (-1)^{\frac{p^s(p-1)r}{q} N}
= (-1)^{\frac{(p-1)r}{q} N},
\eea
while the left-hand side is multiplied by the same factor.
Thus, the theorem is also proved for polynomials (\ref{definitionofT}).

\section{Convergence and analytic continuations}

\subsection{$p$-adic discs} 
For $u\in \matF_p$ let $\tilde{u} \in \matZ_p$ be its Teichmuller lift, i.e., unique lift satisfying $\tilde{u}^p=\tilde{u}$. 
Let us denote
$$
D_{u}=\{a\in \matZ_p: |a-\tilde{u}|_p<1\}
$$
These $p$-adic discs give a partition
$$
\matZ_p=\bigcup_{u \in \matF_p}\, D_u.
$$
For a polynomial $B(z) \in \matZ_p[z]$, let us define
$$
D_{B}=\{a\in \matZ_p: |B(a)|_p=1\}
$$
Let $\bar{B}(z)$ denote the projection of $B(z)$ to $\matF_p[z]$, then $D_B$ is a union of discs:
$$
D_B= \bigcup_{{u \in \matF_p,}\atop {\bar{B}(u)\neq 0}}\, D_u.
$$
For any $k$ we have
\bean \label{diskp}
\{a\in \matZ_p: |B(a^{p^k})|_{p}=1\}=D_{B}.
\eean

\subsection{Domain of uniform convergence}
For $p$-adic approximations of vertex function $\T_s(z)$ we define
$$
\frak{D}=\{z\in \matZ_p: |\T_1(z)|_p=1\}
$$
\begin{lem} \label{moduludlem}
 For every $s\geq 1$ and $a\in \frak{D}$ we have
 $
 |\T_s(z)|_{p}= |\T_s(z^p)|_{p}=1.
 $
\end{lem}

\begin{proof}
By Corollary \ref{corrprod} we have:
$$
\T_s(z)=\T_1(z) \T_1(z^p)\dots \T_1(z^{p^{s-1}}) \mod p
$$
The Lemma follows from (\ref{diskp}). 
\end{proof}

\begin{thm} \label{thmconverg}
 The sequence of function $I_s(z)=\T_{s+1}(z)/\T_s(z^p)$, $s=0,1,\dots$ converges uniformly
  on $\frak{D}$ to an analytic $\matZ_p$-valued  function. If $I(z)$ denotes this function, 
  then for any $a\in \frak D$ we have $|I(a)|_{p}=1$.
\end{thm}

\begin{proof}
By Lemma \ref{moduludlem} and equality (\ref{diskp}), for any  $a\in \frak{D}$ we have $|I(a)|_p=1$. By (\ref{dworkid}), $I_s(z)$ is a Cauchy sequence on $\frak{D}$. The theorem follows. 
\end{proof}

\subsection{Analytic continuation}

Now, let us restore both parameters $z_{k,1}=z_1$ and $z_{n-k,1}=z_2$ in (\ref{Lfacts}) and define a polynomial $\hat{\T}_s(z_1,z_2) \in \matZ[z_1,z_2]$ by the same formula
(\ref{definitionofT}) with the same choice of sign. The polynomials $\hat{\T}_s(z_1,z_2)$ are homogeneous, symmetric,
 $\hat{\T}_s(z_1,z_2)=\hat{\T}_s(z_2,z_1)$,
  and have degree $(p^s-1) r k/q$. They are related to the polynomials $\T_s(z)$ which we considered before via
\bean \label{contt}
\hat{\T}_s(z_1,z_2) = z_1^{\frac{(p^s-1)r}{q}k} \, \T_s(z_2/z_1) = z_2^{\frac{(p^s-1)r}{q}k} \, \T_s(z_1/z_2)
\eean
which implies that $z^{\frac{(p^s-1)r}{q}k}\, \T_s(1/z)=\T_s(z)$.
From (\ref{contt}) we also have:
$$
z_1^{\frac{(p-1)}{q} r k}\, \dfrac{\T_{s+1}(z_2/z_1)}{\T_{s}((z_2/z_1)^p)} = z_2^{\frac{(p-1)}{q} r k}\, \dfrac{\T_{s+1}(z_1/z_2)}{\T_{s}((z_1/z_2)^p)}\,.
$$
Passing to the limit $s\to \infty$ and using Theorem  \ref{thmconverg}
we obtain:

\begin{thm} \label{ancontfunc}
  We have 
  $$
  z^{\frac{(p-1)}{q} r k} \, I(1/z)= I(z)$$ 
  if $z, 1/z \in \frak{D}.$
\end{thm}

 In particular, let $\tilde u$ be the Teichmuller lift of an element $u\in \matF_p$, then $I(\tilde u) = I(1/\tilde u)$ if $k$ is even or $k$ is odd and $u$ is a quadratic residue, or $I(\tilde u) = - I(1/\tilde u)$ if $k$ is odd and $u$ is a quadratic non-residue.

 \begin{example}
 Consider the  hypergeometric function
$$
F(z) = {}_{n-1}F_n\Big(\frac rq, \dots, \frac rq; 1,\dots,1;z\Big).
$$ 
According to our theorems, the function $I(z) =F(z)/F(z^p)$ defined in a neighborhood of $z=0$ as the
ratio of hypergeometric power series can be $p$-adically 
analytically  continued to the corresponding domain $\frak D$ and satisfies  there the identity
$I(z) = z^{(p-1)r k/q}I(1/z)$. Notice that the same function defined over
complex numbers does not have such a relation. 

 \end{example}

\appendix

 \section{Remarks on polynomial superpotential $\Phi_1(x,z)$ and $\F_p$-points}


\subsection{The case $k=1$ and $\om=1/2$} 

Recall that the vertex function in this case is given by the integral (\ref{tpninteg}), 
\bean \label{tpninteg}
\Ver(z)=\dfrac{\alpha}{(2 \pi i)^{n-1}}\oint_{\gamma}\, \dfrac{ dx_{1,1} \wedge \dots \wedge dx_{n-1,1}}{y} \,
\eean
where 
\bean \label{hypersurf}
 y^{2}=x_{1,1}\dots x_{n-1,1} (x_{2,1}-x_{1,1}) \dots (x_{n-1,1}-x_{n-2,1}) (1-x_{1,1}) (z-x_{n-1,1}).
 \eean

 Let $p$ be an odd prime and  $z \in \matF_p$\,.\, Denote  by $N(z)$ 
  the number of $\matF_p$-points on the (singular) affine  hypersurface (\ref{hypersurf}).
  Define $\T_1(z)$ as the coefficient of $\prod_{i=1}^{n-1}x_{i,1}^{p-1}$ in $\Phi_1(x,z)$, which differs by the
 factor $(-1)^{(n-1)(p-1)/2}$ from \eqref{definitionofT}.

\begin{thm}
 For  $z \in \matF_p$\,,\, we have
 $ N(z) =  (-1)^{n-1} \T_1(z) \mod p.$
\end{thm}

\begin{proof}
Let $P(x,z)$ denote  the right-hand side of (\ref{hypersurf}).
If $P(x,z)=0$ for some $x\in \matF_p^{n-1}$,
then
we have one solution of (\ref{hypersurf}) given by the point $(x,0)$. 
If $P(x,z)\neq 0$, then $P(x,z)^{\frac{p-1}{2}} =\Phi_1(x,z) =\pm 1$. By Euler's criterion the 
equation has two solutions $(\pm y,t)$ if $\Phi_1(t,z) = 1$ and no solutions if  $\Phi_1(t,z) = -1$. 
We conclude that, the total number of solutions
is 
$$
N(z)=\sum\limits_{t \in \matF_p^{n-1}}\, (1+\Phi_1(t,z)) 
$$
which  modulo $p$ equals:
$$
N(z)=\sum\limits_{t \in \matF_p^{n-1}}\, \Phi_1 (t,z).
$$
Expanding $\Phi_1(x,z)$  into a sum of monomials we obtain
$$
\Phi_1(x,z) = \sum\limits_{m \in \matN^{n-1}}\, c_m(z) x^m 
$$
where  $x^m=x_{1,1}^{m_1}\dots x_{n-1,1}^{m_{n-1}}$.
Thus
$$
N(z)=\sum\limits_{m \in \matN^{n-1}}\, c_m(z)\, \Big( \sum\limits_{t \in \matF_p^{n-1}}\, t^m \Big).
$$
To compute the sum in the brackets we note that
$$
\sum\limits_{t \in \matF_p^{n-1}}\, t^m =
\left\{\begin{array}{ll}
(p-1)^{n-1}, & \textrm{if $(p-1)| m_i$ for all $i$}\\
0, & \textrm{otherwise}.
\end{array}\right. 
$$
By Lemma \ref{lemmaforT}, the only monomial $x^m$ in $\Phi_1(x,z)$,\,
 such that $(p-1)|m_i$ for all $i$,\, is the monomial
$$
x^{p-1} = \prod\limits_{i} x_{i,1}^{p-1} .
$$
Thus, we conclude that
$$
N(z)=(p-1)^{n-1} c_{p-1}(z) = (-1)^{n-1}\T_{1}(z) 
$$
where the last equality is again modulo $p$. 
 \end{proof}

\subsection{The case $k=1$, $n=2$  and $0<r<q$} 

Recall that $p=\ell q+1$. 
Denote by  $N(z)$ the number of $\F_p$-points on the affine curve
\bean 
\label{curve}
 y^{q}=x_{1,1}^{q-r}(1-x_{1,1})^r (z-x_{1,1})^r.
 \eean
 Denote by $P(x_{1,1},z)$ the right-hand side in \eqref{curve}. 
Recall the first polynomial superpotential
\bea
\Phi_1(x_{1,1},z) = P(x_{1,1},z)^{(p-1)/q}\,,
\eea
see \eqref{polsuper}.
Define $\T_1(z)$ as the coefficient of $x_{1,1}^{p-1}$ in $\Phi_1(x,z)$, which differs by the
 factor $(-1)^{(n-1)(p-1)r/q}$ from \eqref{definitionofT}.

\vsk.3>
We relate the number of points $N(z)$ to the constant term of the ``Fourier expansion'' of the element
$-T_1(z)\in \F_p$ as follows.

Define the nonnegative integer 
\bea
M(z) =\big\vert\{ t\in \F_p\mid \Phi_1(t,z)=1\}\big\vert\,.
\eea

\begin{lem}
\label{lem 1}
 For  $z \in \matF_p-\{0,1\}$\,,\, we have the following equality in $\Z$\,:
 \bean
\label{m}
N(z) = 3 +qM(z)\,.
\eean
\end{lem}

\begin{proof}
Let $\theta\in \F_p^\times $ be a generator.  An element $a\in \F_p^\times$ is a $q$th power, 
if and only if $a=\theta^{q m}$ for some $m$. If $a=\theta^{qm}$, then the equation 
$y^q=a$ has exactly $q$ distinct solutions  $y_i = \theta^{m + \ell i}$, $i =0, 1, \dots, q-1$.
In particular,  the $q$th roots of unity are the elements
\bean
\label{zeta}
\zeta^i=\theta ^{\ell i}, \qquad i=0,1,\dots,q-1.
\eean
These remarks show that  for $t\in\F-\{0,1,z\}$, the equation 
\bean
\label{1}
y^q= P(t,z)
\eean
has solutions  if and only if $\Phi_1(t,z) =1$. Moreover, if $\Phi_1(t,z) =1$, 
then equation \eqref{1} has $q$ distinct solutions.
This reasoning proves equation \eqref{m} in which the summand 3 corresponds to the three points
$\{0,1,z\}$.
\end{proof}

 We have the formula
\bean
\label{T}
- T_1(z) = \sum_{t\in\F_{p}} \Phi_1(t,z)   \pmod{p}.
\eean
Notice that for any  $t\in\F_p-\{0,1,z\}$, the value $\Phi_1(t,z)$ is a $q$th root of unity.
Replacing  each $\Phi_1(t,z)$ in \eqref{T}
by the corresponding root $\zeta^i$ we obtain
the equation
\bean
\label{zt}
- T_1(z) = \sum_{i=0}^{q-1} A_i \,\zeta^i \pmod{p},
\eean
where $A_i$ are non-negative integers with 
\bea
\sum_{i=0}^{q-1}A_i =p-3, \qquad A_0=M(z)\,.
\eea  

\begin{cor}
Consider this presentation \eqref{zt}
of $-T_1(z)$ as a sum of $q$th roots, then we have the following equality in $\Z\,:$
\bean
\label{TN}
N(z) = 3 + q A_0\,.
\eean
\end{cor}


\begin{rem}
The periods of the curve 
$y^q = x^{q-r}(x-1)^r(x-z)^r$
satisfy  the hypergeometric differential equation whose holomorphic solution at $z=0$ is
the hypergeometric function
\bean
{}_2F_1\left(\frac rq, \frac rq; 1;z\right).
\eean
Let $\on{Fr}(u)$ denote the $2\times 2$-matrix of the Frobenius structure for this differential equation.
For $z\in\F_p-\{0,1\}$, consider the  polynomial
\bea
\det\left(x+\on{Fr}(t_z)\right) = x^2+L_z x+A_z\,
\eea
where $t_z$ is the Teichmuller representative of $z\in\F_p$\, and $L_z, \,A_z\in\Z_p$\,.
\smallskip

For $i = 0,1,\dots,q-1$, denote by $\om^i\in\Z_p$ the Teichmuller representative of
the element  $ \zeta^i\in \F_p$ defined in \eqref{zeta}.
We expect the following identity in $\Z_p$\,:
\bean
\label{identity}
-L_z \,= \,\sum_{i=0}^{q-1} A_i\, \om^i\,
\eean
where the nonnegative integers $A_i$ are defined in formula 
\eqref{zt}.

Computer calculations support this statement for  $(r,q, p)$ equal to 
 \bea 
 (1,3,7),\quad (1,3,19), \quad(4,5,11), \quad (5,6,13) \,.
\eea

\end{rem}

\end{document}